\newif\ifJFP
\JFPfalse
\ifJFP
\documentclass{jfp1}
\else
\documentclass[fullpage]{article}
\fi
\usepackage{graphicx}
\usepackage{url}
\usepackage{mathrsfs}
\usepackage[all]{xy}
\usepackage{qsymbols}
\usepackage{enumerate}
\usepackage{prooftree}
\ifJFP
\else
\newcommand{\shortcite}[1]{\cite{#1}}
\fi

\ifJFP\usepackage{minted}\else
\usepackage{listings}
 \usepackage{color}\fi

 \ifJFP \else 
 \newenvironment{proof}[1]{\begin{quotation}\noindent\textsf{Proof:} #1}%
 {\(\Box\)\end{quotation}}
 \fi
 \newtheorem{fact}{Fact}
 \newtheorem{prop}{Proposition}
 \newtheorem{theorem}{Theorem}
 \newtheorem{conj}{Conjecture}

 \newcommand{\ie}{i.e.,~}
 \newcommand{\nat}{\ensuremath{\mathbb{N}}}
 \newcommand{\Ss}{\ensuremath{\mathbb{S}}}
 \newcommand{\St}{\ensuremath{\widetilde{\mathbb{S}}}}

 \newcommand{\Vv}{\ensuremath{{\mathbb{S}_{\infty}}}}

 \newcommand{\Var}[1]{\underline{\mathsf{#1}}}
 
 \newcommand{\nbct}[1]{T_{0,#1}}
 \newcommand{\nbt}[1]{T_{\infty,#1}}

 \newqsymbol{`e}{\varepsilon}
 
 \newcommand{\C}{\mathscr{C}}
 \newcommand{\PP}{\mathbb{P}}
 \newcommand{\Me}{\mathbb{E}}

 \newcommand{\Soo}{\ensuremath{S_{\infty}}}

 \definecolor{rouge}{rgb}{0.6,0,0}
 \newcommand{\rouge}[1]{{\color{rouge} #1}}
 \newcommand{\bl}[1]{\textcolor{blue}{#1}}
 \newif\ifcomment

 \commenttrue           

 \newcommand{\out}[1]{}

\begin{document}
 \ifJFP
\else
 \title{Counting and Generating Terms \\in the Binary Lambda
   Calculus\\\emph{(Extended version)}} 
 \author{Katarzyna Grygiel$^1$\thanks{This work was
     partially supported by the grant 2013/11/B/ST6/00975 founded by the Polish
     National Science Center.}  ~and
   Pierre Lescanne$^{1,2}$\vspace{12pt}\\
   $^1$Jagiellonian University,
   Faculty of Mathematics and Computer Science,\\
   Theoretical Computer Science Department, \\
   ul. Prof. {\L}ojasiewicza 6, 30-348 Krak\'ow, Poland\vspace{10pt}\\
   $^2$University of Lyon, \\
   \'Ecole normale sup\'erieure de Lyon, \\
   LIP (UMR 5668 CNRS ENS Lyon UCBL INRIA)\\
   46 all\'ee d'Italie, 69364 Lyon, France\\
   (email: \texttt{grygiel@tcs.uj.edu.pl,pierre.lescanne@ens-lyon.fr})
 }
 \fi

 \maketitle

 \begin{abstract}
   In a paper entitled \emph{Binary lambda calculus and combinatory logic}, John Tromp
   presents a simple way of encoding lambda calculus terms as binary sequences. In
   what follows, we study the numbers of binary strings of a given size that represent
   lambda terms and derive results from their generating functions, especially that
   the number of terms of size $n$ grows roughly like $1.963447954\ldots^n$.  In a
   second part we use this approach to generate random lambda terms using Boltzmann samplers.
  \vspace{6pt}
 \end{abstract}

 \section{Introduction}

 In recent years growing attention has been given to quantitative research in logic
 and computational models. Investigated objects (e.g., propositional formulae,
 tautologies, proofs, programs) can be seen as combinatorial structures, providing
 therefore the inspiration for combinatorists and computer scientists. In particular,
 several works have been devoted to studying properties of lambda calculus terms. From the
 practical point of view, generation of random $`l$-terms is the core of debugging
 functional programs using random tests~\cite{DBLP:conf/icfp/ClaessenH00} and the
 present paper offers an answer to an open question (see introduction
 of~\cite{DBLP:conf/icfp/ClaessenH00}) since we are able to generate closed typable
 terms following a uniform distribution.  But this work applies beyond $`l$-calculus to
 any system with bound variables, like the first order predicate calculus (quantifiers are
 binders like $`l$) or block structures in programming languages.

 First traces of the combinatorial approach to lambda calculus date back to the work
 of Jue Wang~\shortcite{wang04:_effic_gener_random_progr_their_applic}, who initiated the
 idea of enumerating $`l$-terms. In her report, Wang defined the size of a term as the
 total number of abstractions, applications and occurrences of variables, which
 corresponds to the number of all vertices in the tree representing the given term.

 This size model, although natural from the combinatorial viewpoint, turned out to be
 difficult to handle. The question that arises immediately concerns the number of
 $`l$-terms of a given size. This task has been done for particular classes 
 of terms by Bodini, Gardy, and Gittenberger
 \shortcite{DBLP:journals/combinatorics/BodiniGGJ13} and
 Lescanne~\shortcite{DBLP:journals/tcs/Lescanne13}.

 The approach applied in the latter paper has been extended
 in~\cite{DBLP:journals/jfp/GrygielL13} by the authors of the current paper to the
 model in which applications and abstractions are the only ones that contribute to the
 size of a $`l$-term. The same model has been studied by David et
 al.~\shortcite{DBLP:journals/corr/abs-0903-5505}, where several properties satisfied
 by random $`l$-terms are provided.

 When dealing with the two described models, it is not difficult to define recurrence
 relations for the number of $`l$-terms of a given size. Furthermore, by applying standard
 tools of the theory of generating functions one obtains generating functions that are
 expressed in the form of infinitely nested radicals. Moreover, the radii of convergence are
 in both cases equal to zero, which makes the analysis of those functions very
 difficult to cope with.

 In this paper, we study the binary encoding of lambda calculus introduced in
 \cite{DBLP:conf/dagstuhl/Tromp06}. This representation results in another size
 model. It comes from the binary lambda calculus defined by Tromp, in which he builds a
 minimal self-interpreter of lambda calculus\footnote{An alternative to universal
 Turing machine.} as a basis of algorithmic complexity theory~\cite{LiVitanyi}.  
 Such a binary approach is more realistic from the functional programming viewpoint.
 Indeed, for compiler builders it
 is counter-intuitive to assign the same size to all the variables, because in the
 translation of a program written in \textsf{Haskell}, \textsf{Ocaml} or \textsf{LISP}
 variables are put in a stack.  A~variable deep in the stack is not as easily
 reachable as a variable shallow in the stack.  Therefore the weight of the former
 should be larger than the weight of the latter. Hence it makes sense to associate a
 size with a variable proportional to its distance to its binder.  %
 When we submitted~\cite{DBLP:journals/jfp/GrygielL13} to the \emph{Journal of Functional Programming},
 a referee wrote: ``If the authors want to use the de Bruijn representation, another
 interesting experiment could be done: rather than to count variables as size 0, they
 should be counted using their \emph{unary} representation.  This would penalize deep
 lexical scoping, which is not a bad idea since 'local' terms are much easier to
 understand and analyze than deep terms''.  %
 In this model, recurrence relations for the number of terms of a given size are built
 using this specific notion of size.  From that, we derive corresponding generating
 functions defined as infinitely nested radicals. However, this time the radius of
 convergence is positive and enables a further analysis of the functions. We are able
 to compute the asymptotics of the number of all (not necessarily closed) terms and we
 also prove an upper bound of the asymptotics of the number of closed ones. Moreover,
 we define an unranking function, \ie a generator of terms from their indices from
 which we derive a uniform generator of random $`l$-terms (general and typable) of a
 given size.  This allows us to provide outcomes of computer experiments in which we
 estimate the number of simply typable $`l$-terms of a given size.  

 Recall that Boltzmann samplers are programs for efficient generation of random
 combinatorial objects. Based on generating functions, they are parameterized by the
 radius of convergence of the generating function.  In addition to a more realistic
 approach of the size of the $`l$-terms, binary lambda calculus terms are associated
 with a generating function with a positive radius of convergence, which allows us to
 build a Boltzmann sampler, hence a very efficient way to generate random
 $`l$-terms. In Section~\ref{sec:Boltzmann} and Section~\ref{sec:lambda} we introduce
 the notion of Boltzmann sampler and we propose a Boltzmann sampler for $`l$-terms
 together with a \textsf{Haskell} program.

 A version~\cite{DBLP:journals/corr/GrygielL14} of the first part of this paper
 was presented at the \emph{25th International Conference on Probabilistic,
   Combinatorial and Asymptotic Methods for the Analysis of Algorithms}. 

 \section{Lambda calculus and its binary representation}

 In order to eliminate
 names of variables from the notation of $`l$-terms, de Bruijn introduced an
 alternative way of representing equivalent terms.

 Let us assume that we are given a countable set $\{ \Var{1}, \Var{2}, \Var{3}, \ldots \}$, 
  elements of which are called de Bruijn indices. We define de Bruijn 
 terms (called terms for brevity) in the following way:
 \begin{enumerate}[(i)]
 \item each de Bruijn index $\Var{i}$ is a term,
 \item if $M$ is a term, then $(\lambda M)$ is a term (called an abstraction),
 \item if $M$ and $N$ are terms, then $(MN)$ is a term (called an application).
 \end{enumerate}

 For the sake of clarity, we will omit the outermost parentheses. Moreover, we sometimes omit 
 other parentheses according to the convention that application associates to the left, and 
 abstraction associates to the right. Therefore, instead of $(MN)P$ we will write $MNP$, and instead of $\lambda (\lambda M)$ we will write $\lambda \lambda M$.

 Given a term $\lambda N$ we say that the $`l$ encloses all indices occurring in the term $N$. Given a term $M$, we say that an occurrence of an index $\Var{i}$ in the term $M$ is \emph{free} in $M$ if the number of $`l$'s in $M$ enclosing the occurrence of $\Var{i}$ is less than $i$. Otherwise, we say the given occurrence of $\Var{i}$ is bound by the $i$-th lambda enclosing it. A term $M$ is called closed if there are no free occurrences of indices.

 For instance, given a term $\lambda \lambda \Var{1} ( \lambda \Var{1} \ \Var{4})$, the first occurrence of $\Var{1}$ is bound by the second lambda, the second occurrence of $\Var{1}$ is bound by the third lambda, and the occurrence of $\Var{4}$ is free. Therefore, the given term is not closed.

 Following John Tromp, we define the binary representation of de Bruijn indices in the
 following way:
 \begin{eqnarray*}
   \widehat{\lambda M} &=& 00\widehat{M},\\
   \widehat{MN} &=& 01\widehat{M}\widehat{N},\\
   \widehat{\Var{i}} &=& 1^i0.
 \end{eqnarray*}
 However, notice that unlike Tromp~\shortcite{DBLP:conf/dagstuhl/Tromp06} and
 Lescanne~\shortcite{LescannePOPL94}, we start the de Bruijn indices at $1$ like de
 Bruijn~\shortcite{deBruijn72}.  Given a de Bruijn term, we define its size as the length of the
 corresponding binary sequence, \ie
 \begin{eqnarray*}
   |\Var{n}| &=& n + 1,\\
   |`l M | &=& |M| + 2,\\
   |M\,N| &=& |M| + |N| + 2.
 \end{eqnarray*}

 For instance, the de Bruijn term $\lambda \lambda \Var{1} ( \lambda \Var{1} \ \Var{4})$
 is represented by the binary sequence \linebreak $0000011000011011110$ and hence
 its length is $19$.

 In contrast to models studied previously, the number of all (not necessarily closed)
 \mbox{$`l$-terms} of a given size is always finite. This is due to the fact that the size of
 each variable depends on the distance from its binder.

 \section{Combinatorial facts}
 In order to determine the asymptotics of the number of all/closed $`l$-terms of a
 given size, we will use the following combinatorial notions and results.

 We say that a sequence $(F_n)_{n \geq 0}$ is of
 \begin{itemize}
 \item order $G_n$, for some sequence $(G_n)_{n \geq 0}$ (with $G_n\neq 0$), if
   \[ \lim_{n \to \infty} F_n/G_n = 1, \] and we denote this fact by $F_n \sim G_n$;
 \item exponential order $A^n$, for some constant $A$, if
   \[ \limsup_{n \to \infty} |F_n|^{1/n} = A, \] and we denote this fact by $F_n
   \bowtie A^n$.
 \end{itemize}

 Given the generating function $F(z)$ for a sequence $(F_n)_{n \geq 0}$, we write
 $[z^n]F(z)$ to denote the $n$-th coefficient of the Taylor expansion of $F(z)$,
 therefore $[z^n] F(z) = F_n$.

 The theorems below (Theorem IV.7 and Theorem VI.1 of \cite{flajolet08:_analy_combin})
 serve as powerful tools that allow us to estimate coefficients of certain functions that
 frequently appear in combinatorial considerations.

 \begin{fact}
   If $F(z)$ is analytic at $0$ and $R$ is the modulus of a singularity nearest to the
   origin, then
   \[ [z^n]F(z) \bowtie (1/R)^n.\]
 \end{fact}

 \begin{fact}\label{fact:asym_exp}
   Let $\alpha$ be an arbitrary complex number in $\mathbb{C}\setminus
   \mathbb{Z}_{\leq 0}$. The coefficient of $z^n$ in
   \[ f(z) = (1-z)^{\alpha} \] admits the following asymptotic expansion:
   \begin{eqnarray*} [z^n]f(z) &\sim& \frac{n^{\alpha - 1}}{\Gamma (\alpha)} \left( 1
       + \frac{\alpha (\alpha - 1)}{2n} + \frac{`a(`a-1)(`a-2)(3`a-1)}{24n^2} + O
       \left( \frac{1}{n^{3}} \right) \right),
   \end{eqnarray*}
   where $\Gamma$ is the Euler Gamma function.
 \end{fact}

 \section{The sequences $S_{m,n}$}

 Let us denote the number of $`l$-terms of size $n$ with at most $m$ distinct free
 indices by $S_{m,n}$.

 First, let us notice that there are no terms of size $0$ and $1$. Let us consider a
 $`l$-term of size $n+2$ with at most $m$ distinct free indices. Then we have one of
 the following cases.
 \begin{itemize}
 \item The term is a de Bruijn index $\Var{n+1}$, provided $m$ is greater than or
   equal to $n+1$.
 \item The term is an abstraction whose binary representation is given by
   $00\widehat{M}$, where the size of $M$ is $n$ and $M$ has at most $m+1$ distinct
   free variables.
 \item The term is an application whose binary representation is given by
   $01\widehat{M}\widehat{N}$, where $M$ is of size $i$ and $N$ is of size $n-i$, with
   $i \in \{ 0,\ldots,n \}$, and each of the two terms has at most $m$ distinct free variables.
 \end{itemize}

 This leads to the following recursive formula\footnote{Given a predicate $P$,
   $[P(\vec{x})]$ denotes the Iverson symbol, i.e., $[P(\vec{x})] = 1$ if $P(\vec{x})$
   and $[P(\vec{x})] = 0$ if $\neg P(\vec{x})$.}:
 \begin{eqnarray}
   S_{m,0} &=& S_{m,1} ~=~ 0,\label{eq:Smn}\\
   S_{m,n+2} &=& [m \ge n+1] + S_{m+1,n} + \sum_{k=0}^n S_{m,k} S_{m,n-k}.\label{eq:Smn2}
 \end{eqnarray}
 The sequence $(S_{0,n})_{n \geq 0}$, \ie the sequence of numbers of closed $`l$-terms of size $n$,
 can be found in the \emph{On-line Encyclopedia of Integer Sequences} under the number
 \textbf{A114852}. Its first $20$ values are as follows:
 \[ 0,\ 0,\ 0,\ 0,\ 1,\ 0,\ 1,\ 1,\ 2,\ 1,\ 6,\ 5,\ 13,\ 14,\ 37,\ 44,\ 101,\ 134,\
 298,\ 431.\]
 More values are given in Figure~\ref{fig:nb_typ}.
 The values of $S_{m,n}$ can be
 computed by the function we call $\mathsf{tromp}$ given in Figure~\ref{fig:tromp}.
 \begin{figure}[!th]
     \hrule\medskip
     \begin{lstlisting}
 -- Iverson symbol
 iv b = if b then 1 else 0

 -- Tromp size
 a114852Tab :: [[Integer]]
 a114852Tab = [0,0..] : [0,0..] : [[iv (n - 2 < m) + 
                                a114852Tab !! (n-2) !! (m+1) + 
                                s n m 
                               | m <- [0..]] | n <- [2..]]
   where s n m  = let ti = [a114852Tab !! i !! m | i <- [0..(n-2)]] in 
           sum (zipWith (*) ti (reverse ti))

 tromp m n = a114852Tab !! n !! m
   \end{lstlisting}
     \hrule\medskip
   \caption{The function \textsf{tromp} computing the $S_{m,n}$}
 \label{fig:tromp}
 \end{figure}

 Now let us define the family of generating functions for sequences $(S_{m,n})_{n \geq  0}$:
 \begin{eqnarray*}
   \Ss_m(z) &=& \sum_{n=0}^{\infty} S_{m,n}\, z^n.
 \end{eqnarray*}

 Most of all, we are interested in the generating function for the number of closed
 terms, \ie
 \begin{eqnarray*}
   \Ss_0(z) &=& \sum_{n=0}^{\infty} S_{0,n}\, z^n.
 \end{eqnarray*}

 Applying the recurrence on $S_{m,n}$, we get
 \begin{eqnarray*}
   \Ss_m(z) &=& z^2 \sum_{n= 0}^\infty S_{m,n+2} z^n\\
   &=& z^2 \sum_{n= 0}^\infty [m \ge n+1]z^n + z^2 \sum_{n= 0}^\infty S_{m+1,n}\,z^n + z^2 \sum_{n= 0}^\infty\sum_{k=0}^n
   S_{m,k} S_{m,n-k}\,z^n\\
   &=&  z^2 \sum_{k=0}^{m-1} z^k + z^2 \Ss_{m+1}(z) + z^2 \Ss_m(z)^2\\
   &=& \frac{z^2\,(1-z^m)}{1-z}  + z^2 \Ss_{m+1}(z) + z^2 \Ss_m(z)^2.
 \end{eqnarray*}

 Solving the equation
 \begin{eqnarray}\label{eq:Szu}
   z^2 \Ss_m(z)^2 -  \Ss_m(z) +\frac{z^2\,(1-z^m)}{1-z}  + z^2 \Ss_{m+1}(z) = 0
 \end{eqnarray}
 gives us
 \begin{eqnarray}\label{eq:sm}
   \Ss_m(z) \ =\ \frac{1 - \sqrt{1 - 4z^4 \left(\frac{1-z^m}{1-z}  +  \Ss_{m+1}(z)\right)}}{2 z^2}.
 \end{eqnarray}

 This means that the generating function $\Ss_m(z)$ is expressed by means of
 infinitely many nested radicals, a phenomenon which has already been encountered in
 previous research papers on enumeration of $`l$-terms, see e.g.,
 \cite{bodini11:_lambd_bound_unary_heigh}. However, in Tromp's binary lambda calculus
 we are able to provide more results than in other representations of $`l$-terms.

 First of all, let us notice that the number of $`l$-terms of size $n$ has to be
 less than $2^n$, the number of all binary sequences of size $n$. This means that in
 the considered model of $`l$-terms the radius of convergence of the generating
 function enumerating closed $`l$-terms is positive (it is at least $1/2$), which
 is not the case in other models, where the radius of convergence is equal to zero.

 \section{The number of all $`l$-terms}

 Let us now consider the sequence enumerating all binary $`l$-terms, \ie including
 terms that are not closed. Let $S_{\infty,n}$ denote the number of all such terms of
 size $n$. Repeating the reasoning from the previous section, we obtain the following
 recurrence relation:

 \begin{eqnarray*}
   S_{\infty,0} &=& S_{\infty,1} ~=~ 0,\\
   S_{\infty,n+2} &=& 1 + S_{\infty,n} + \sum_{k=0}^n S_{\infty,k} S_{\infty,n-k}.
 \end{eqnarray*}

 The sequence $(S_{\infty,n})_{n`:\nat}$ can be found in \emph{On-line Encyclopedia of
   Integer Sequences} with the entry number \textbf{A114851}. Its first $20$ values
 are as follows:
 \[ 0,\ 0,\ 1,\ 1,\ 2,\ 2,\ 4,\ 5,\ 10,\ 14,\ 27,\ 41,\ 78,\ 126,\ 237,\ 399,\ 745,\
 1292,\ 2404,\ 4259.\]
 More values are given in Figure~\ref{fig:nb_typ}.

 Obviously, we have $S_{m,n} \le S_{\infty,n}$ for every $m,n \in \nat$. Moreover,
 $\displaystyle \lim_{m\to \infty} S_{m,n} = S_{\infty,n}$.

 Let $\Vv(z)$ denote the generating function for the sequence
 $(S_{\infty,n})_{n`:\nat}$, that is
 \[\Vv(z) = \sum_{n=0}^{\infty}S_{\infty,n} z^n.\]
 Notice that for $m \geq n-1$ we have $S_{m,n}=S_{\infty,n}$. Therefore,
 \[\Vv(z) = \sum_{n=1}^{\infty}S_{n,n}z^{n},\]
 which yields that $[z^n]\Ss_{n,n}=[z^n]\Ss_{\infty,n}$. Furthermore, $\displaystyle
 \Vv(z) = \lim_{m \to \infty} \Ss_m(z)$ for all ${z \in (0,\rho)}$, where $\rho$ is 
 the dominant singularity of $\Ss_\infty(z)$.

 \begin{theorem}
   The number of all binary $`l$-terms of size $n$ satisfies
   \[ S_{\infty,n} \sim \rho^{-n} \cdot \frac{C}{n^{3/2}},\] where $\rho \doteq
   0.509308127$ and $C \doteq 1.021874073$.
 \end{theorem}

 \begin{proof}
   The generating function $\Vv(z)$ fulfills the equation
   \[\Vv(z) = \frac{z^2}{1-z} + z^2 \Vv(z) + z^2 \Vv(z)^2.\]

   Solving the above equation gives us
   \[\Vv(z) = \frac{ (1-z)(1-z^2) - \sqrt{(1-z)(1 - z - 2\,z^2 + 2\,z^3 - 3\,z^4 - z^5)}}{2z^2(1 - z)}.\]

   The dominant singularity of the function $\Vv(z)$ is given by the root of smallest
   modulus of the polynomial
   \[R_\infty(z) = 1 - z - 2\,z^2 + 2\,z^3 - 3\,z^4 - z^5.\]

   The polynomial has three real roots: \[0.509308127\ldots, \quad -0.623845142\ldots, \quad -3.668100004\ldots,\] 
   and two complex ones that are approximately equal to $0.4 + 0.8i$ and $0.4 - 0.8i$.

   Therefore, $\rho \doteq 0.509308127$ is the singularity of $\Vv$ nearest to the
   origin.  Let us write $\Vv(z)$ in the following form:
   \[\Vv(z) = \frac{ 1-z^2 - \sqrt{\rho (1-\frac{z}{\rho}) \cdot Q(z)}}{2z^2},\] 
   where $Q(z)$
   is a rational function defined for all $|z| \le \rho$.

   We get that the radius of convergence of $\Vv(z)$ is equal to $\rho$ and its
   inverse $\frac{1}{`r} \doteq 1.963447954$ gives the growth of
   $S_{\infty,n}$. Hence, $S_{\infty,n} \bowtie (1/\rho)^n$.

   Fact \ref{fact:asym_exp} allows us to determine the subexponential factor of the
   asymptotic estimation of the number of terms. Applying it, we obtain
   \[ [z^n]\Vv(z) = \rho^{-n} [z^n]\Vv(\rho z) 
                  \sim \rho^{-n} [z^n] \frac{- \sqrt{1-z} \cdot \sqrt{\rho Q(\rho z)}}{2\rho^2z^2} 
                  \sim \rho^{-n} \cdot \frac{n^{-3/2}}{\Gamma (-\frac{1}{2})} \cdot \widetilde{C}, \] 
   where the constant $\widetilde{C}$ is given by
   \[\widetilde{C} = \frac{- \sqrt{\rho \cdot Q(\rho)}}{2 \rho ^{2}}
   \doteq -0.288265354.\] Since $\displaystyle \frac{\widetilde{C}}{\Gamma
     (-\frac{1}{2})} \doteq 1.021874073$, the theorem is proved.
 \end{proof}

 \section{The number of closed $`l$-terms}

 \begin{prop}
   Let $\rho_m$ denote the dominant singularity of $\Ss_m(z)$. Then for every natural
   number $m$ we have
   \[\rho_m = \rho_0,\]
   which means that all functions $\Ss_m(z)$ have the same dominant singularity.
 \end{prop}

 \begin{proof}
   First, let us notice that for every $m,n \in \nat$ we have $S_{m,n} \le
   S_{m+1,n}$. This means that the radius of convergence of the generating function
   for the sequence $(S_{m,n})_{n \in \nat}$ is not smaller than the radius of
   convergence of the generating function for $(S_{m+1,n})_{n \in \nat}$. Therefore,
   for every natural number $m$, we have
   \[ \rho_m \ge \rho_{m+1} .\]

  Additionally, from Equation~(\ref{eq:sm}) we see that every singularity of
   $\Ss_{m+1}(z)$ is also a singularity of $\Ss_m(z)$. Hence, the dominant singularity
   of $\Ss_{m}(z)$ is less than or equal to the dominant singularity of
   $\Ss_{m+1}(z)$, \ie we have
   \[ \rho_m \le \rho_{m+1} .\]

   These two inequalities show that dominant singularities of all functions $\Ss_m(z)$
   are the same. In particular, for every $m$ we have $\rho_m = \rho_0$.
 \end{proof}
 \begin{prop}
   The dominant singularity of \ $\Ss_0(z)$ is equal to the dominant singularity of
   $\Vv(z)$, \ie
   \[\rho_0 \ = \ \rho \ \doteq \ 0.509308127. \]
 \end{prop}

 \begin{proof}
   Since the number of closed binary $`l$-terms is not greater than the number of all
   binary terms of the same size, we conclude immediately that $\rho_0 \geq \rho$.

   Let us now consider the functionals $`F_{\infty}$ and $`F_m$ for every $m \in \nat$.  
   By Equation~(\ref{eq:sm}), for every $m$ the functional $`F_m$ applied to  $\Ss_{m+1}$
   gives us $\Ss_{m}$, while  $`F_{\infty}$ is the limit of the sequence $(`F_m)_{m \in \nat}$:
   \begin{eqnarray*}
     `F_m(F) &=& \frac{1-\sqrt{1- 4z^4(\frac{1-z^m}{1-z} + F)}}{2z^2},\\
     `F_{\infty}(F) &=& \frac{1-\sqrt{1- 4z^4(\frac{1}{1-z} + F)}}{2z^2}.
   \end{eqnarray*}

   In particular, when $m=0$, we have
   \begin{eqnarray*}
     `F_0(F) &=&  \frac{1-\sqrt{1 - 4z^4F}}{2z^2}.
   \end{eqnarray*}
   By Equation~(\ref{eq:sm}) and the definition of $`F_m$, we have
   \begin{eqnarray*}
     \Ss_m(z) &=& `F_m(\Ss_{m+1}(z)).
   \end{eqnarray*}

   The $`F_m$'s and $`F_{\infty }$ are monotonic over functions over $(0,1)$, which
   means that for every $z \in (0,1)$ we have
   \begin{eqnarray*}
     F(z) \le G(z) &"=>"&`F_m(F(z)) \le `F_m(G(z)),\\
     F(z) \le G(z) &"=>"&`F_{\infty}(F(z)) \le `F_{\infty}(G(z)).
   \end{eqnarray*}

   For each $m \in \nat$, let us consider the function $\St_m(z)$ defined as the fixed
   point of $`F_m$. In other words, $\St_m(z)$ is defined as the solution of the
   following equation:
   \begin{eqnarray*}
     \St_m(z) &=& `F_m(\St_m(z)).
   \end{eqnarray*}

   Notice that $S_{m,n} \le S_{m+1,n} \le S_{\infty,n}$, for the reasons that
   given a size $n$, there are less trees with at most $m$ free variables than trees
   with at most $m+1$ free variables and less trees with at most $m+1$ free variables
   than trees with any numbers of free variables.  For $z \in (0,\rho)$ we can claim
   that $\Ss_m(z) \le\Ss_{m+1}(z) \le \Ss_\infty(z)$. Applying $`F_m$ to the first inequality,
   we obtain, for $z \in (0,\rho)$,
   \begin{eqnarray*}
     `F_m(\Ss_m(z)) &\le & \Ss_m(z)
   \end{eqnarray*}
 Then we get 
 \begin{displaymath}
    `F_m^{k+1}(\Ss_m(z)) \le`F_m^k(\Ss_m(z))\le ... \le`F_m(\Ss_m(z)) \le  \Ss_m(z)
 \end{displaymath}
 and since
 \begin{displaymath}
   \St_m(z) = \lim_{k"->"\infty} `F_m^k(\Ss_m(z)) = \inf_{k`:\nat} `F_m^k(\Ss_m(z)) 
 \end{displaymath}
 we infer
   \begin{displaymath}
     \St_m(z) \le  \Ss_m(z) \ \le \ \Vv(z). \label{eq:tilde}
   \end{displaymath}
   Since $\St_m(z)$ satisfies
   \begin{eqnarray*}
     2z^2\St_m(z) &=& 1 -  \sqrt{1 - 4z^4 \Big( \frac{1-z^m}{1-z} + \St_m(z) \Big)},
   \end{eqnarray*}
   we get
   \begin{eqnarray*}
     z^2\St_m^2(z) - (1-z^2) \St_m(z) + \frac{z^2(1-z^m)}{1-z} = 0.
   \end{eqnarray*}

   The discriminant of this equation is:
   \begin{eqnarray*}
     `D_m &=& (1-z^2)^2 - \frac{4z^4(1-z^m)}{1-z}.
   \end{eqnarray*}

   The values for which $`D_m = 0$ are the singularities of $\St_m(z)$. Let us denote
   the main singularity of $\St_m(z)$ by $`s_m$. From Equation~(\ref{eq:tilde}) we see
   that
   \begin{eqnarray*}
     `s_m\ge `r_m \ge `r. \label{eq:sigma_m_rho_m_rho}
   \end{eqnarray*}

   The value of $`s_m$ is equal to the root of smallest modulus of the following
   polynomial:
   \[ P_m(z) \ := \ (z-1)`D_m \ = \ 4z^4(1-z^m) - (1-z)^3 (1+z)^2 .\]

   In the case of the function $\St_\infty(z)$, we get the polynomial
   \begin{eqnarray*}
     P_\infty(z)&=& - 1 + z + 2\,z^2 - 2\,z^3 + 3\,z^4 + z^5 \ =\ - R_\infty(z),
   \end{eqnarray*}
   whose root of smallest modulus is the same as for $R_\infty(z)$, hence it is equal to $`r$.

   Now, let us show that the sequence $(\sigma_m)_{m \in \nat}$ of roots of smallest
   modulus of polynomials $P_m(z)$ is decreasing and that it converges to
   $`r$. As a hint, Figure~\ref{fig:roots} illustrates plots of polynomials $P_m$'s (for several
   values of $m$) in the interval $[0.3,\, 1]$.  It shows the roots of the $P_m$'s at
   the intersection of the curves and of the horizontal axis, between $`r$ (for
   $P_{\infty}$) and $1$ (for $P_0$).
     \begin{figure}[t!h]
       \centering
       \includegraphics[scale=.5]{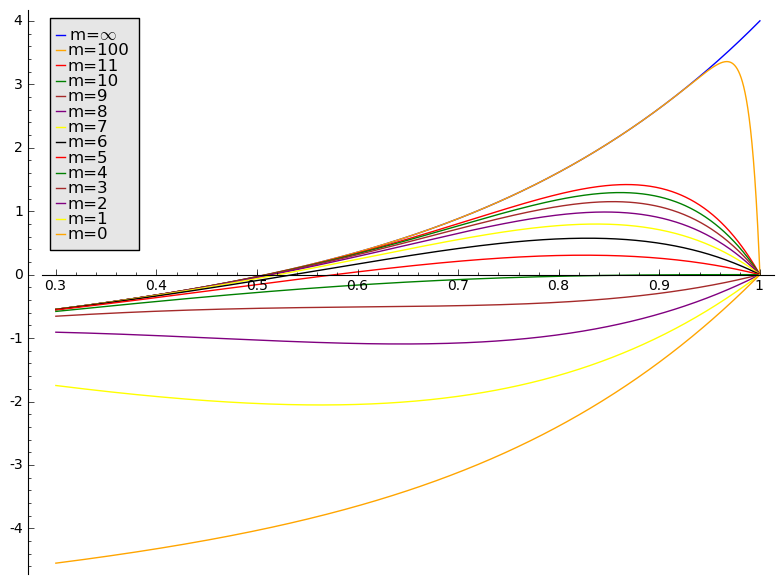}
       \caption{Plots of the $P_m$'s.  The top curve is $P_{\infty}$, below there is
         $P_{100}$, then $P_{10}$, $P_9$ etc. until $P_1$ and $P_0$.}
       \label{fig:roots}
     \end{figure}

   Notice that $P_m(z) = P_\infty(z) - 4z^{m + 4}$. Given a value $`z$ such that
   $`r<`z<1$ (for instance $`z=0.8$), $P_m(z)$ converges uniformly to $P_\infty(z)$ in
   the interval $[0,`z]$. Therefore $`s_m "->" `r$ when $m"->"\infty$. By $`s_m \geq
   `r_m \geq `r$, we get $`r_m "->" `r$, as well. Since all the $`r_m$'s are equal, we
   obtain that $`r_m =`r$ for every natural $m$.
 \end{proof}

 The number of closed terms of a given size cannot be greater than the number of all
 terms. Therefore, we immediately obtain what follows.

 \begin{theorem}
   The number of closed binary $`l$-terms of size $n$ is of exponential order
   $(1/\rho)^n$, \ie
   \[ S_{0,n} \bowtie 1.963448\ldots^n.\]
 \end{theorem}


 Figure \ref{fig:Smn_ran_n32} shows values $S_{m,n} \cdot `r^n \cdot n^{3/2}$ for a
 few initial values of $m$ and $n$ up to $600$  and  allows us to state the following conjecture.

 \begin{figure}[t!h]
   \centering
   \includegraphics[scale=.5]{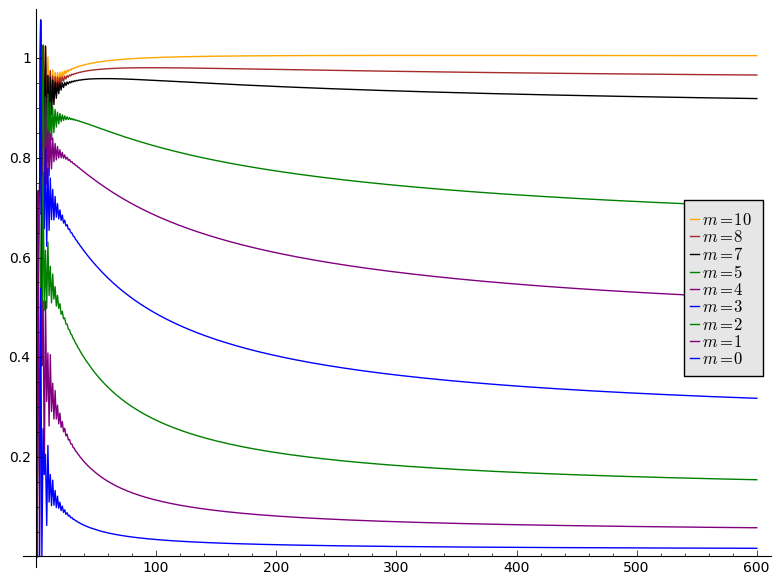}
   \caption{$S_{m,n} `r^n n^{3/2}$ up to $n=600$ for $m = 0$ (bottom) to $10$ (top)}
   \label{fig:Smn_ran_n32}
 \end{figure}

 \begin{conj}
   For every natural number $m$, we have
   \[ S_{m,n} \sim o \big( 1.963448\ldots^n \cdot n^{-3/2} \big). \]
 \end{conj}

 \section{Unrankings}
 \label{sec:unrank}

 The recurrence relation (\ref{eq:Smn2}) for $S_{m,n}$ allows us to define the
 function generating $`l$-terms. More precisely, we construct bijections $s_{m,n}$,
 called \emph{unranking} functions, between all non-negative integers not greater than
 $S_{m,n}$ and binary $`l$-terms of size $n$ with at most $m$ distinct free
 variables~\cite{ranking}. This approach is also known as the \emph{‘recursive
   method’}, originating with Nijenhuis and Wilf~\shortcite{nijenhuis78:_combin} (see
 especially Chapter~13).  

 Let us recall that for $n\ge2$ we have, by (\ref{eq:Smn2}),
 \begin{eqnarray*}
   S_{m,n} &=& [m \ge n-1] + S_{m+1,n-2} + \sum_{j=0}^{n-2} S_{m,j} S_{m,n-2-j}.
 \end{eqnarray*}
 The encoding function $s_{m,n}$ takes an integer $k \in \{1,\ldots,S_{m,n}\}$ and
 returns the term built in the following way. 
 \begin{itemize}
 \item If $m \geq n-1$ and $k$ is equal to $S_{m,n}$, the function returns the string
   $1^{n-1} 0$.
 \item If $k$ is less than or equal to $S_{m+1,n-2}$, then the corresponding term is
   in the form of abstraction $00 \widehat{M}$, where $\widehat{M}$ is the
   value of the unranking function $s_{m+1,n-2}$ on $k$.
 \item Otherwise, i.e., if $k$ is greater than $S_{m+1,n-2}$ and less than $S_{m,n}$ for
   $m\ge n+1$ or less than or equal to $S_{m,n}$ for $m < n+1$, the corresponding
   term is in the form of application $01 \widehat{M} \widehat{N}$. In
   order to get strings $\widehat{M}$ and $\widehat{N}$, we compute the
   maximal value $\ell \in \{0,\ldots,n-2\}$ for which
 $$k - S_{m+1,n-2} = \Big( \sum_{j=0}^{\ell-1} S_{m,j}S_{m,n-2-j} \Big) +r \qquad \textrm{with~} \ r \le
 S_{m,\ell} S_{m,n-2-\ell}.$$%
 The strings $\widehat{M}$ and $\widehat{N}$ are the values
 $s_{m,\ell}(k')$ and $s_{m,n- 2-\ell}(k'')$, respectively, where $k'$ is the integer
 quotient upon dividing $r$ by $S_{m,n-2-\ell}$, and $k''$ is the remainder.
 \end{itemize}
 Notice that in this definition, extremal values of $k$ are considered first. 
 Namely, first the maximal value $S_{m,n}$ (for $m\ge n-1$) is considered, then values from the set $\{1,\ldots,S_{m+1,n-2}\}$ are taken into account, and finally, in the third case, the remaining values.

 In Figure~\ref{fig:unrankT-prog} the reader may find a \textsf{Haskell} definition of 
 the data type \texttt{Term} and
 a program~\cite{jones03:_haskel} for computing the values $s_{m,n}(k)$.  
 In this program, the function $s_{m,n}(k)$ is written as \texttt{unrankT m n k} and the
 sequence $S_{m,n}$ is written as \texttt{tromp m n}.


 \begin{figure}[!th]
   \begin{center}
     \hrule\medskip
         \begin{lstlisting}
 data Term = Index Int
           | Abs Term
           | App Term Term

 unrankT :: Int -> Int -> Integer -> Term
 unrankT m n k
   | m >= n - 1 && k == (tromp m n) = Index (n - 1)
   | k <= (tromp  (m+1) (n-2)) = Abs (unrankT (m+1) (n-2) k)
   | otherwise = unrankApp (n-2) 0 (k - tromp  (m+1) (n-2))
     where unrankApp n j r
             | r <=  tmjtmnj  = let (dv,rm) = (r-1) `divMod` tmnj
                                in App (unrankT m j (dv+1)) 
                                       (unrankT  m (n-j) (rm+1))
             | otherwise = unrankApp n (j + 1) (r -tmjtmnj) 
             where tmnj = tromp m (n-j)
                   tmjtmnj = (tromp m j) * tmnj 
     \end{lstlisting}
  \hrule
   \end{center}
   \caption{The data type \texttt{Term} and a program for computing values of the function $s_{m,n}$}
   \label{fig:unrankT-prog}
 \end{figure}

 \section{Number of typable terms}\label{sec:typable}

 The unranking function allows us to traverse all the closed terms of size $n$ and to
 filter those that are typable (see~\cite{hindley97:_basic_simpl_theor})
 in order to count them and similarly to traverse all the terms of size $n$ to count
 those that are typable.  

 The comparison of the numbers of $`l$-terms and the numbers of typable $`l$-terms
 is presented in Figure~\ref{fig:nb_typ}. From left to right:
 \begin{enumerate}
 \item the numbers $S_{0,n}$ of closed terms (typable and untypable) of size $n$,
 \item the numbers $\nbct{n}$ of closed typable terms of size $n$,
 \item the numbers $S_{\infty,n}$ of all terms (typable and untypable) of size $n$,
 \item the numbers $\nbt{n}$ of all typable terms of size $n$.
 \end{enumerate}

 \begin{figure}
   \begin{scriptsize}
     \begin{center}
       \begin{math}
         \begin{array}[t]{r l}
           \mathbf{n} & S_{0,n}\\\hline
           0&0\\
           1&0\\
           2&0\\
           3&0\\
           4&1\\
           5&0\\
           6&1\\
           7&1\\
           8&2\\
           9&1\\
           10&6\\
           11&5\\
           12&13\\
           13&14\\
           14&37\\
           15&44\\
           16&101\\
           17&134\\
           18&298\\
           19&431\\
           20&883\\
           21&1361\\
           22&2736\\
           23&4405\\
           24&8574\\
           25&14334\\
           26&27465\\
           27&47146\\
           28&89270\\
           29&156360\\
           30&293840\\
           31&522913\\
           32&978447\\
           33&1761907\\
           34&3288605\\
           35&5977863\\
           36&11148652\\
           37&20414058\\
           38&38071898\\
           39&70125402\\
           40&130880047\\
           41&242222714\\
           42&452574468\\
           43&840914719\\
           44&1573331752\\
           45&2933097201\\
           46&5495929096
         \end{array}
         \qquad
         \begin{array}[t]{r l}
           \mathbf{n} & \nbct{n}\\\hline
           0&0\\
           1&0\\
           2&0\\
           3&0\\
           4&1\\
           5&0\\
           6&1\\
           7&1\\
           8&1\\
           9&1\\
           10&5\\
           11&4\\
           12&9\\
           13&13\\
           14&23\\
           15&29\\
           16&67\\
           17&94\\
           18&179\\
           19&285\\
           20&503\\
           21&795\\
           22&1503\\
           23&2469\\
           24&4457\\
           25&7624\\
           26&13475\\
           27&23027\\
           28&41437\\
           29&72165\\
           30&128905\\
           31&227510\\
           32&405301\\
           33&715078\\
           34&1280127\\
           35&2279393\\
           36&4086591\\
           37&7316698\\
           38&13139958\\
           39&23551957\\
           40&42383667\\
           41&76278547\\
           42&137609116\\
           43&248447221\\
           44&449201368\\
           45&812315229\\
           46&1470997501\\
         \end{array}
       \qquad\qquad
       \begin{array}[t]{r l}
         \mathbf{n} & S_{\infty,n}\\\hline
         0&0\\
         1&0\\
         2&1\\
         3&1\\
         4&2\\
         5&2\\
         6&4\\
         7&5\\
         8&10\\
         9&14\\
         10&27\\
         11&41\\
         12&78\\
         13&126\\
         14&237\\
         15&399\\
         16&745\\
         17&1292\\
         18&2404\\
         19&4259\\
         20&7915\\
         21&14242\\
         22&26477\\
         23&48197\\
         24&89721\\
         25&164766\\
         26&307294\\
         27&568191\\
         28&1061969\\
         29&1974266\\
         30&3698247\\
         31&6905523\\
         32&12964449\\
         33&24295796\\
         34&45711211\\
         35&85926575\\
         36&161996298\\
         37&305314162\\
         38&576707409\\
         39&1089395667\\
         40&2061428697\\
         41&3901829718\\
         42&7395529009\\
         43&14 023 075 765\\
         44&26620080576\\
         45&50556677634\\
         46&96 108 150 292
       \end{array}
 \qquad
 \begin{array}[t]{r l}
   \mathbf{n} & \nbt{n}\\[.2pt]\hline
   0&0\\
   1&0\\
   2&1\\
   3&1\\
   4&2\\
   5&2\\
   6&3\\
   7&5\\
   8&8\\
   9&13\\
   10&22\\
   11&36\\
   12&58\\
   13&103\\
   14&177\\
   15&307\\
   16&535\\
   17&949\\
   18&1645\\
   19&2936\\
   20&5207\\
   21&9330\\
   22&16613\\
   23&29921\\
   24&53588\\
   25&96808\\
   26&174443\\
   27&316267\\
   28&572092\\
   29&1040596\\
   30&1888505\\
   31&3441755\\
   32&6268500\\
   33&11449522\\
   34&20902152\\
   35&38256759\\
   36&70004696\\
   37&128336318\\
   38&235302612\\
   39&432050796\\
   40&793513690\\
   41& 1459062947\\
   42& 2683714350\\
   43 & \textit{unknown}\\
   44 & \textit{unknown}\\
   45 & \textit{unknown}\\
   46 & \textit{unknown}\\
 \end{array}
         \end{math}
     \end{center}
   \end{scriptsize}

   \caption{Numbers of terms and numbers of typable terms}
   \label{fig:nb_typ}
 \end{figure}
 In particular, let us notice that $S_{0,n}$ and $T_{0,n}$ are the same up to $n=8$,
 where we meet the smallest untypable closed term namely $`l \Var{1}\,\Var{1}$.  Similarly,
 $S_{\infty,n}$ and $T_{\infty,n}$  are the same up to $n=6$, where we meet the
 smallest untypable term, namely $\Var{1}\,\Var{1}$.  Values  $T_{\infty,43}$,  $T_{\infty,44}$,  
 $T_{\infty,45}$ and $T_{\infty,46}$ are not available since they require too many computations, 
 between $14$ millions and $96$ millions of $`l$-terms have to be checked for
 typability in each case.  Paul Tarau~\shortcite{tarau15:_type_gener_lambd_terms} gives
 a Prolog implementation and applies it to the generation of typed $`l$-terms.

 Thanks to the unranking function, we can build a \emph{uniform generator of
   $`l$-terms} and, using this generator, we can build a \emph{uniform generator of simply
   typable $`l$-terms}, which sieves the uniformly generated terms
 through a program that checks their typability (see for
 instance~\cite{DBLP:journals/jfp/GrygielL13}).  
 This way, it is possible to generate
 typable closed terms uniformly up to size $450$\footnote{Tromp constructed a self-interpreter
 (which is not typable) for the $`l$-calculus of size $210$.}.


 \section{Boltzmann samplers}
 \label{sec:Boltzmann}

 In this section we present the basic ideas related to Boltzmann models, which combined
 with the theory of generating functions allow us to develop efficient algorithms for 
 generating random $`l$-terms. A thorough and clear overview of Boltzmann samplers, 
 including many examples, can be found in~\cite{DBLP:journals/cpc/DuchonFLS04}. For readers
 not acquainted with the theory, we provide necessary notions and constructions.

 Let $\C$ be a combinatorial class, i.e., a set of combinatorial objects endowed with 
 a size function $|\cdot| \colon \C \to \nat$ such that there are finitely many elements 
 of size $n$ for every $n \in \nat$.
 Let $C_n$ denote the cardinality of the subset of $\C$ consisting of elements 
 of size $n$. Furthermore, let $C(z)$ denote the generating functions associated with the sequence 
 $(C_n)_{n \in \nat}$, which means that
 \begin{displaymath}
   C(z) = \sum_{n=0}^{\infty} C_n z^n.
 \end{displaymath}
 Notice that
 \begin{displaymath}
   C(z) = \sum_{`g`:\C} z^{|`g|}.
 \end{displaymath}

 Given a positive real $x \in {\mathbb R}_+$, we define a \emph{Boltzmann model} for the class 
 $\C$ as the probability distribution that assigns to every element $\gamma \in \C$ 
 a~probability 
 \begin{displaymath}
   \PP_{\C,x}(`g) = \frac{1}{C(x)} \cdot x^{|`g|}.
 \end{displaymath}
 This is a probability since 
 \begin{displaymath}
   \sum_{`g`:\C} \PP_{\C,x}(`g) = \sum_{`g`:\C}\frac{1}{C(x)} \cdot x^{|`g|} = 1.
 \end{displaymath}
 The role of $x$ will become clear later on, but for now we may consider $x$ as a
 parameter used to ``tune'' the sampler, that is to center the mean value around a
 chosen number.  In other words, if we want to set an expected mean value, we have to
 compute the proper value of $x$.  In order the probability $ \PP_{\C,x}(`g)$ to be
 well-defined, we assume the values of $x$ to be taken from the interval
 $(0,\rho_{\C})$, where $\rho_{\C}$ denotes the radius of convergence of $C(z)$.
 Provided $C(z)$ converges at $\rho_{\C}$, we may also consider the case $x =
 \rho_{\C}$.

 The size of an object in a Boltzmann model is a random variable $N$. The \emph{Boltzmann sampler} 
 for a class $\C$ and a parameter $x$ is a random object generator, which draws from the 
 class $\C$ an object of size $n$ with probability
 \begin{displaymath}
   \PP_{x}(N=n) = \frac{C_n x^n}{C(x)}.
 \end{displaymath}
 This is indeed a well-defined probability since
 \begin{displaymath}
   \sum_{n\ge0} \PP_{x}(N=n) = \frac{1}{C(x)}\sum_{n\ge 0} C_n x^n = 1.
 \end{displaymath}

 When generating random objects, we require either the size to be a fixed value~$n$
 or, in order to increase the efficiency of the generation process, we admit some 
 flexibility on the size. In other words, we want the objects to be generated in some 
 cloud around a given size $n$ so that the size $N$ of the objects lies in some interval 
 $(1-\varepsilon)n \le N \le (1+\varepsilon)n$ for some factor $\varepsilon>0$ called 
 a \emph{tolerance}. Such a method is called \emph{approximate-size} uniform random generation.
 What we want to preserve is the uniformity of the distribution among objects of the same size, 
 \ie we want all objects of the same size to be drawn with the same probability.

 The random variable $N$ has \emph{a first moment} and \emph{a second moment}
 \cite{DBLP:journals/cpc/DuchonFLS04}:
 \begin{displaymath}
   \Me_x(N) = x\frac{C'(x)}{C(x)} \qquad \Me_x(N^2) = \frac{x^2 C''(x) + x C'(x)}{C(x)},
 \end{displaymath}
 and a \emph{standard deviation}:
 \begin{eqnarray*}
   `s_{x}(N) &=& \sqrt{\Me_x(N^2) - \Me_x(N)^2} \\
   &=& \sqrt{\frac{x^2 C''(x) + x
       C'(x)}{C(x)} - x^2\frac{C'(x)^2}{C(x)^2}}.
 \end{eqnarray*}

 In the case of approximate-size generation, in order to maximize chances of drawing 
 an object of a desired size, say $n$, we need to choose a proper value of the parameter $x$.
 It turns out that the best value of $x$ is such for which $\Me_x(N)=n$ (for a detailed study
 see \cite{DBLP:conf/analco/FlajoletFP07}). Given size $n$, we will denote by $x_n$ the 
 value of the parameter chosen in such a way. Moreover, if $n$ tends to infinity, then
 $x_n$ tends to $\rho_C$ (see Appendix).

 \subsection{Design of Boltzmann generator}
 \label{sec:design}

 A Boltzmann generator for a class $\C$ is built according to a recursive
 specification of the class~$\C$. Since we are interested in designing a Boltzmann
 sampler for binary $`l$-terms, we present the way of defining samplers for classes
 which are specified by means of the following recursive constructions: disjoint
 unions (data type \texttt{\rouge{Either}}~a~ b), products (data type
 \texttt{\rouge{Pair}}) and sequences (data type \textsf{\rouge{List}}).  First we
 assume a monad \texttt{\rouge{Gen}} defined from the monad
 \texttt{\rouge{State}} of the \textsf{Haskell} library by
 \begin{lstlisting}
   type Gen  = State StdGen
 \end{lstlisting}
 where \texttt{\rouge{StdGen}} is the type of standard random generators.  For the following we
 assume a function \texttt{rand :: \rouge{Gen Double}} that generates a random double
 precision real in the interval $(0,1)$ together with an update of the random
 generator. In our case it is defined in Figure~\ref{fig:rand}.
 \begin{figure}[!th]
   \begin{lstlisting}
 rand :: Gen Double
 rand = do generator <- get
           let (value, newGenerator) = randomR (0,1) generator
           put newGenerator
           return value
   \end{lstlisting} 
   \caption{The function \textsf{rand}}
 \label{fig:rand}
 \end{figure}

 \subsection{Disjoint union}
 \label{sec:union}

 Let \texttt{a} and \texttt{b} be two types (corresponding to combinatorial
 classes $\mathscr{A}$ and $\mathscr{B}$).  A generator \texttt{\bl{genEither}} for
 the disjoint union takes a double precision number for the Bernoulli choice and two
 objects of type \texttt{\rouge{Gen} a} and \texttt{\rouge{Gen} b} and returns an
 object of type \texttt{\rouge{Gen} (\rouge{Either} a b)}.  If we define a new class as \texttt{c =
   \rouge{Either} a b} corresponding to the class $\C$  with the size function inherited
 from classes $\mathscr{A}$ and $\mathscr{B}$, then $C_n = A_n + B_n$ and $C(z) = A(z)
 + B(z)$.  The probability of drawing an object $\gamma \in \C$ equals
 \begin{displaymath}
   \PP_{\C,x}(`g`:\mathscr{A}) = \frac{A(x)}{C(x)}, \qquad \PP_{\C,x}(`g`:\mathscr{B)} = \frac{B(x)}{C(x)}.
 \end{displaymath}

 A generator for the disjoint union, i.e., a Bernoulli variable, may have the
 following type:
 \begin{lstlisting}
 genEither::Double -> (Gen a) -> (Gen b) -> (Either a b -> c) -> Gen c
 \end{lstlisting}
 and then it is given by the 
 \textsf{Haskell} function: 
 \begin{lstlisting}
 genEither p ga gb caORb = do
     x <- rand
     if x < p then do ga' <- ga
                      return (caORb $ Left ga')
              else do gb' <- gb
                      return (caORb $ Right gb')
 \end{lstlisting}
 Notice the type of \textsf{\bl{genEither}} which assumes that \textsf{\bl{genEither}}
 takes a number, two monad values \textsf{\rouge{Gen}~a} and \textsf{\rouge{Gen}~b}
 (which can be seen as pairs of a random generator and a value of type \textsf{a} and
 \textsf{b} respectively), and a continuation \textsf{c} of type
 \textsf{\rouge{Either} a b} and returns a value of the monad \textsf{\rouge{Gen} c}.
 Similar frames will appear in the programs describing other generators.

 \subsection{Cartesian product}

 Given classes $\mathscr{A}$ and $\mathscr{B}$, let $\mathscr{C}$ be the class 
 defined as their Cartesian product, \ie $\mathscr{C} = \mathscr{A} \times \mathscr{B}$.
 Let \texttt{a} and \texttt{b} be Haskell types corresponding to classes 
 $\mathscr{A}$ and $\mathscr{B}$. Then the type of the class $\mathscr{C}$ is $\mathtt{(a,b)}$.
 The size of an object $\gamma = \langle \alpha,\beta \rangle \in \C$ equals the sum 
 of sizes $|\alpha| + |\beta|$.  In more concrete terms, if an object is the pair of 
 an object of size $p$ and an object of size $q$, then its size is $p+q$. 
 Hence, the generating functions satisfy the equation $C(z) = A(z) \cdot B(z)$, since
 \begin{displaymath}
   C(z) = \sum_{\langle`a,`b\rangle `: \mathscr{A} \times \mathscr{B}} z^{|`a|+|`b|}.
 \end{displaymath}
 The probability of drawing $`g=\langle`a,`b\rangle`:\C$ is equal to
 \begin{displaymath}
   \PP_{\C,x}(`g) = \frac{x^{|`g|}}{C(x)} =\frac{x^{|`a+`b|}}{A(x)\cdot B(x)} =  \frac{x^{|`a|}}{A(x)}\cdot \frac{x^{|`b|}}{B(x)}.
 \end{displaymath}
 In this case the Boltzmann sampler is as follows:

 \begin{minipage}{1.0\linewidth}
 \begin{lstlisting}
 genPair :: (Gen a) -> (Gen b) -> (a -> b -> c) -> (Gen c)
 genPair ga gb caANDb = do
      ga' <- ga
      gb' <- gb
      return (caANDb ga' gb')
 \end{lstlisting}
 \end{minipage}

 \section{Boltzmann samplers for $`l$-terms}
 \label{sec:lambda}

 Let us consider the equation involving the generating function for all $`l$-terms:
 \[\Soo(z) = \frac{z^2}{1-z} + z^2 \Soo(z) + z^2 \Soo(z)^2.\]
 It is derived from the description of the set $\mathscr{S}_\infty$
 of $`l$-terms as:
 \[\mathscr{S}_\infty = \mathscr{D} + `l \mathscr{S}_\infty + \mathscr{S}_\infty\,\mathscr{S}_\infty.\]
 That means that the set of $`l$-terms $\mathscr{S}_\infty$ has three components: the
 first component $\mathscr{D}$ corresponds to de Bruijn indices, the second component
 $`l\mathscr{S}_\infty$ corresponds to abstractions, the third component
 $\mathscr{S}_\infty\mathscr{S}_\infty$ corresponds to applications.  We build a
 sampler of random terms based on this trichotomy.  In \textsf{Haskell} this
 corresponds to a data type \texttt{\rouge{Term}} defined in Figure~\ref{fig:unrankT-prog}.
 Since there are three components in the union,
 the value $\mathtt{p}$ which we
 considered in $\mathtt{genEither}$ will be replaced by two values $\mathtt{p1}$ and
 $\mathtt{p2}$.  First we describe in \textsf{Haskell} a~function corresponding to $\Soo(z)$:
   \begin{lstlisting}
 sInfinity z = num z / den z
    where num z = z^3 - z^2 - z + 1 - sqrt(sq z)
          den z = 2*z*z*(1 - z)
          sq z = z^6 + 2*(z^5) - 5*(z^4) + 4*(z^3) - z^2 - 2*z + 1
   \end{lstlisting}
 and two functions:
 \begin{lstlisting}
 p1 x = x*x / (1-x) / sInfinity x
 p2 x = p1 x + x^2
 \end{lstlisting}
 Using \textsf{Sage} we computed the values:
 \begin{footnotesize}
   \begin{displaymath}
     x_{100} \,=\, 0.5092252666102192\quad
     x_{600} \,=\, 0.5093058457062517\quad 
     x_{1000} \,=\, 0.5093073063214039
   \end{displaymath}
 \end{footnotesize}
 which correspond to the values of the parameter $x$ appropriate for an expected value
 $\Me_{x_i}(N)$ equal to $i=100$, $i=600$, and $i=1000$, respectively.  In
 other words if the values $x_{100}$, $x_{600}$ and $x_{1000}$ are passed to the
 sampler, it will generate objects with average size $100$, $600$, and $1000$,
 respectively.  They are obtained by solving in $x$ the equations
 \begin{eqnarray*}
   \Me_x(N) &=& 100,\\
   \Me_x(N) &=& 600,\\
   \Me_x(N) &=& 1000,
 \end{eqnarray*}
 in which $C(x)$ is replaced by $\Soo(x)$.

 \subsection{General samplers of $`l$-terms}
 The values of the probabilities for a given $x$ are
 \begin{itemize}
 \item $p_v(x) = \frac{x^2}{(1-x) S_\infty(x)}$ for variables,
 \item $p_{abs}(x) = x^2$ for abstractions,
 \item $p_{app}(x) = x^2 S_\infty (x)$ for applications.
 \end{itemize}
 We get the following \textsf{Haskell} function which selects among
 \texttt{\rouge{Index}}, \texttt{\rouge{Abs}} and \texttt{\rouge{App}}
 \begin{lstlisting}
 genTermGeneric :: Double -> Gen Int -> Gen Term
 genTermGeneric x gi  = do
     p <- rand
     if p < p1 x
     then do i <- genIntGeneric x
             return (Index i)
     else if p < p2 x
          then do t <- genTermGeneric x gi
                  return (Abs t)
          else genPair (genTermGeneric x gi) (genTermGeneric x gi) App
 \end{lstlisting}
 Notice the call to the function
 \begin{lstlisting}
 genIntGeneric :: Double -> Gen Int
 genIntGeneric x = do
     p <- rand
     if p < x then do n <- genIntGeneric x
                      return (n+1)
                else return 1
 \end{lstlisting}
 which is used to generate random de Bruijn indices.

 \subsection{Samplers for large $`l$-terms}
 \begin{figure}[h!t]
   \begin{lstlisting}
 rho :: Double
 rho = 0.509308127024237357194177485
 rhosquare = rho * rho
 p1rho = (1 - rhosquare) / 2
 p2rho = p1rho + rhosquare

 genTerm ::  Gen Int -> Gen Term
 genTerm gi = do
     p <- rand
     if p < p1rho then do i <- genInt
                          return (Index i)
                   else if p < p2rho
                        then do t <- genTerm gi
                                return (Abs t)
                        else genPair (genTerm gi) (genTerm gi) App

 genInt :: Gen Int
 genInt = do
     p <- rand
     if p < rho then do n <- genInt
                        return (n+1)
                else return 1
   \end{lstlisting}
   \caption{The function \texttt{genTerm}}
 \label{fig:genTerm}
 \end{figure}
 As discussed in the previous section, in order to generate random
 large $\lambda$-terms, i.e., \mbox{$`l$-terms} with average size $\infty$, we set the value
 of $x$ to $\rho = 0.5093081270242373\ldots$, which we call \texttt{rho} in
 \textsf{Haskell}. Its square is \ifJFP
\else \[`r^2 = 0.25939476825293667\ldots.\]\fi  Notice that since
 $`r$ is a root of the polynomial below the square root, $S_\infty(`r) =
 \frac{1-`r^2}{2`r^2}$.  The values of the probabilities for selecting among
 variables, abstractions and applications are:
 \begin{itemize}
 \item $p_v(`r) = \frac{2`r^4}{(1-`r)(1-`r^2)}$ \ for variables,
 \item $p_{abs}(`r) = `r^2$ \ for abstractions,
 \item $p_{app}(`r) = \frac{1-`r^2}{2}$ \ for applications.
 \end{itemize}
 Let us simplify $\frac{2`r^4}{(1-`r)(1-`r^2)}$ into $\frac{1-`r^2}{2}$ by computing
 the difference:
 \begin{eqnarray*}
   \frac{2`r^4}{(1-`r)(1-`r^2)} - \frac{1-`r^2}{2} &=& \frac{4`r^4 -
     (1-`r^2)^2(1-`r)}{2(1-`r)(1-`r^2)}\\
   &=& \frac{`r^5 + 3`r^4 -2`r^3 + 2`r^2 + `r -1}{2(1-`r)(1-`r^2)} \,=\, 0.
 \end{eqnarray*}
 Therefore to generate random terms of mean size going to infinity we get the results
 \begin{itemize}
 \item $p_v(`r) = \frac{1-`r^2}{2} \approx 0.3703026$ \ for variables,
 \item $p_{abs}(`r) = `r^2 \approx 0.25939476$ \ for abstractions,
 \item $p_{app}(`r) = \frac{1-`r^2}{2} \approx 0.3703026$ \ for applications.
 \end{itemize}
 We build the function \texttt{genTerm} which generates random terms and the function
 \texttt{genInt} which generates integers necessary for the de Bruijn indices (see
 Figure~\ref{fig:genTerm}).
 The list 
 \ifJFP\begin{small}\else \begin{tiny}\fi
   \begin{displaymath}
     60,5,3,3,6,19,8,7,728,3753,12,15,3733,93,4,3,4,4,13,137,6,18,372,50,25,43140,8,5,3,6
   \end{displaymath}
 \ifJFP\end{small}\else \end{tiny}\fi
 is the list of term sizes generated by \texttt{\bl{genTerm}} when the
 seeds of the random generator are $0,1, 2, \ldots$ up to $30$.  In the same list of term
 sizes the $50^{th}$ element is $127\,358$ and the $51^{st}$ element is
 $4\,379\,394$, showing that generating a random term of size more than four million
 is easy. 

 Assume now that we want to generate terms that are below a certain \texttt{uplimit}, as
 required by practical applications.  The function called \texttt{\bl{ceiledGenTerm}}
 is almost the same as \texttt{\bl{genTerm}}, except that when the up limit is passed
 it returns \texttt{\rouge{Nothing}}.  Therefore the type of \texttt{\bl{ceiledGenTerm}}
 differs from  \texttt{\bl{genTerm}} type in the sense that it takes a
 \texttt{\rouge{Gen (Maybe Term)}} (instead of a \texttt{\rouge{Gen Term}}) and returns a
 \texttt{\rouge{Gen (Maybe Term)}} (instead of a \texttt{\rouge{Gen Term}}).  A
 Boltzmann sampler \texttt{\bl{ceiledGenTerm}} for large $`l$-terms of size limited by
 \texttt{uplimit} is given in Figure~\ref{fig:ceiledGenTerm}.
 \begin{figure*}
   \centering
   \begin{scriptsize}
   \begin{lstlisting}
 ceiledGenTerm :: Int -> Gen Int -> Gen (Maybe Term)
 ceiledGenTerm uplimit gi = do
     p <- rand
     if p < p1rho 
     then do  -- generate an index 
          i <- genInt
          return $ if i < uplimit then Just (Index i) else Nothing
     else if p < p2rho 
          then do -- generate an abstraction
               mbt <- ceiledGenTerm uplimit gi
               return $ case mbt of
                        Just t -> if 2 + size t <= uplimit
                                  then Just (Abs t)
                                   else Nothing
                        Nothing -> Nothing
           else do -- generate an application
                mbt1 <- ceiledGenTerm uplimit gi
                mbt2 <- ceiledGenTerm uplimit gi
                return $ case mbt1 of
                         Just t1 -> case mbt2 of
                                    Just t2 -> if 2 + size t1 + size t2 <= uplimit
                                               then Just (App t1 t2)
                                                else Nothing
                                    Nothing -> Nothing 
   \end{lstlisting}
 \end{scriptsize}  
   \caption{Boltzmann sampler for large $`l$-terms}
   \label{fig:ceiledGenTerm}
 \end{figure*}

 Suppose now that we want to generate terms within a size interval, i.e., with an
 up limit and a down limit.  By definition, \texttt{\bl{ceiledGenTerm}} generates terms
 within an up limit.  For terms within the down limit, terms generated by
 \texttt{\bl{ceiledGenTerm}} are filtered so that only terms large enough are
 kept.  Recall that the method is linear in time complexity.  Thus the generation
 of a term of size $100\,000$ takes a few seconds, the generation of a term of size one
 million takes three minutes and the generation of a term of size five million takes
 five minutes.

 To generate large typable $`l$-terms
 we generate $`l$-terms and check their typability. Currently we are able to generate
 random typable $`l$-terms of size $500$.  This outperforms methods based on ranking
 and unranking like the method proposed in~\cite{DBLP:journals/jfp/GrygielL13}. This
 is in particular due to the fact that such methods need to handle numbers of
 arbitrary precision and their random generation, which it not efficient. Indeed
 ranking or unranking requires handling integers with hundred digits or more and performing
 computations on them for their random generations.  On the other hand, Boltzmann
 samplers ignore numbers, go directly toward the terms to be generated and do that efficiently.

 \section{Related works}
 \label{sec:rel-works}

 We look at related works from two perspectives: works on counting $`l$-terms
 and works specifically related to Boltzmann samplers.

 \subsection{Works on counting $`l$-terms}
 \label{sec:works-counting}

 Connected to this work, let us mention papers on counting \mbox{$`l$-terms}
 \cite{DBLP:journals/tcs/Lescanne13,DBLP:journals/jfp/GrygielL13} and on evaluating
 their combinatorial properties, namely
 \cite{bodini11:_lambd_bound_unary_heigh,DBLP:journals/corr/abs-0903-5505,DBLP:journals/combinatorics/BodiniGGJ13,DBLP:journals/tcs/BodiniGJ13}.
 A~comparison of our results with those of \cite{DBLP:journals/jfp/GrygielL13} can be made,
 since \cite{DBLP:journals/jfp/GrygielL13} gives a precise counting of $`l$-terms when
 variables (de Bruijn indices) have size~$0$, yielding sequence \textsf{A220894} in
 the \emph{On-line Encyclopedia of Integer Sequences} for the number of closed terms
 of size $n$. The first fifteen terms are:
 \begin{center}
   \begin{math}
     0,\ 1,\ 3,\ 14,\ 82,\ 579,\ 4741,\ 43977,\ 454283,\ 5159441,\ 63782411,\
     851368766,
   \end{math}

   \begin{math}
     12188927818,\ 186132043831,\ 3017325884473.
   \end{math}
 \end{center}
 If one compares them with the  first fifteen terms of $S_{0,n}$:
 \[ 0,\ 0,\ 0,\ 0,\ 1,\ 0,\ 1,\ 1,\ 2,\ 1,\ 6,\ 5,\ 13,\ 14,\ 37,\]
 one sees that $S_{0,n}$ grows much more slowly than \textsf{A220894}, which is not
 surprising since $S_{0,n}$ grows exponentially, whereas \textsf{A220894} grows
 super-exponentially (the radius of convergence of its generating function is $0$).
 This super-exponential growth and the related $0$ radius of convergence prevent from
 building a Boltzmann sampler.  Moreover, it does not make sense to count all 
 (including open) terms of size
 $n$ when variables have size $0$ for the reason that there are infinitely many such
 terms for each $n$.  Notice that taking the size of variables to be $1$,
 like~\cite{DBLP:journals/tcs/Lescanne13,bodini11:_lambd_bound_unary_heigh}, does not
 make much difference for growth and generation.

 \subsection{Works related to Boltzmann samplers for terms}
 \label{sec:works}

 In the introduction we cited papers that are clearly connected to this work.  In a
 recent work, Bacher et al.~\shortcite{DBLP:journals/corr/BacherBJ14} propose an improved
 random generation of binary trees and Motzkin trees, based on R\'{e}my algorithm~\cite{DBLP:journals/ita/Remy85} (or
 algorithm R in Knuth~\shortcite{KnuthVol4_4}).  They
 propose like R\'{e}my to grow the trees from inside by an operation called grafting.
 It is not clear how this can be generalized to $`l$-terms as one needs ``to find a
 combinatorial interpretation for the holonomic equations [which] is not [...]
 always possible, and even for simple combinatorial objects this is not elementary''
 (Conclusion of \cite{DBLP:journals/corr/BacherBJ14} page~16).

 \section{Conclusion}
 \label{sec:concl}

 We have shown that if the size of a lambda term is yielded by its binary
 representation~\cite{DBLP:conf/dagstuhl/Tromp06}, we get an exponential growth of the
 sequence enumerating $`l$-terms of a given size.  This applies to closed
 $`l$-terms, to $`l$-terms with a bounded number of free variables, and to all
 $`l$-terms of size $n$. Except for the case of all $`l$-terms, the question of
 finding the non-exponential factor of the asymptotic approximation of the numbers
 of those terms is still open. Moreover, we have described unranking functions (recursive methods) 
 for generating $`l$-terms, which allow us to derive tools
 for their uniform generation and for enumeration of typable $`l$-terms.  The
 process of generating random (typable) terms is limited by the performance of the generators
 based on the recursive methods aka unranking since huge numbers are involved.  
 It turns out that implementing Boltzmann samplers, central tools for the uniform 
 generation of random structures such as trees or $`l$-terms, gives significantly better 
 results. There are now two directions for further development: the first one consists in integrating 
 the programs proposed here in actual testers and optimizers~\cite{DBLP:conf/icfp/ClaessenH00} 
 and the second one in extending Boltzmann samplers to other kinds of programs, for instance programs with block
 structures. From the theoretical point of view, more should be known about generating
 functions for \emph{closed $`l$-terms} or \emph{$`l$-terms with fixed bounds on the
 number of free variables}.  Boltzmann samplers should be designed for such terms,
 which requires extending the theory. As concerns combinatorial properties of
 \emph{simply typable $`l$-terms}, many question are left open and seem to be
 hard. Besides, since we are interested in generating typable terms, it is worth
 designing random uniform samplers that deliver typable terms directly.

 \subsection*{Acknowledgements}

 The authors are happy to acknowledge people who commented early versions of this
 paper, and contributed to improve it, especially the editors and the referees of the
 \emph{Journal of Functional Programming} and of the \emph{25th International
   Conference on Probabilistic, Combinatorial and Asymptotic Methods for the Analysis
   of Algorithms}.  The authors would like to give special mentions to Olivier Bodini, Bernhard
 Gittenberger, \'{E}ric Fusy, Patrik Jansson, Marek Zaionc and the participants of
 \emph{the 8th Workshop on Computational Logic and Applications.}


\appendix

\section*{The case $x=`r_\C$: generating objects with mean size $\infty$}
\label{sec:rhoC}

In this section we show that choosing a value $x = `r_\C$ for the parameter of a
sampler yields mean size $\infty$ of the generated objects.

Assume that a generating function we consider is of the form:
\begin{displaymath}
  C(x) = \frac{P_C(x) - \sqrt{Q_C(x)}}{R_C(x)},
\end{displaymath}
where $P_C(x)$, $Q_C(x)$ and $R_C(x)$ are three polynomials and where $`r_\C$ is such that
${Q_C(`r_\C) = 0}$ and where $Q_C(x)> 0$ and $R_C(x) \neq 0$ for $0<x < `r_\C$.  Those properties are
fulfilled by the generating function $S_\infty(x)$.  Indeed,
\begin{eqnarray*}
  P_{\Soo}(z) &=& (1-z)(1-z^2),\\
Q_{\Soo}(z) &=& (1-z)(1 - z - 2\,z^2 + 2\,z^3 - 3\,z^4 - z^5),\\
R_{\Soo} &=& 2z^2(1 - z).
\end{eqnarray*}
Notice that $Q'_C(`r_\C)<0$ in the vicinity of $`r_\C$, i.e., in an interval
$(`r_\C-\varepsilon, `r_\C)$ (because $Q_C(`r_\C) = 0$ and $Q_C(x)>0$ for
$x`:(0,`r_\C]$) and
\begin{displaymath}
  C(`r_\C) = \frac{P_C(`r_\C)}{R_C(`r_\C)}
\end{displaymath}
is finite. On the other hand,
\begin{displaymath}
  C'(x) = \frac{P'_C(x)}{R_C(x)} - \frac{Q'_C(x)}{2\sqrt{Q_C(x)}R_C(x)}  - \frac{(P_C(x) - \sqrt{Q_C(x)}) R'_C(x)}{R_C(x)^2}
\end{displaymath}
shows that
\begin{displaymath}
  \lim_{x"->"`r_\C} C'(x) = \infty.
\end{displaymath}
Hence 
\begin{displaymath}
  \lim_{x"->"`r_\C} \Me_x(N) = \lim_{x"->"`r_\C} \frac{x C'(x)}{C(x)} = \infty.
\end{displaymath}
Therefore, if we choose $x$ to be  $`r_\C$, the size of the generated structures will
be distributed all over the natural numbers, with an infinite average size.
\end{document}